\documentclass[10pt]{article}
\expandafter\let\csname equation*\endcsname\relax
\expandafter\let\csname endequation*\endcsname\relax
\usepackage{amsmath}
\usepackage{amsfonts,array,amssymb,hhline,amsthm,pb-diagram, bm, fancybox}
\usepackage{latexsym,mathdots,mathtools,mathrsfs, ytableau, arydshln, authblk}
\usepackage[a4paper]{geometry}
\usepackage[all]{xy}

\def\qed{\hfill $\Box$}
\theoremstyle{definition}

\newtheorem{defi}{Definition}[section]
\newtheorem{prop}{Proposition}[section]
\newtheorem{lemm}{Lemma}[section]

\newtheorem{ex}{Example}[section]

\newcommand{\wt}[1]{\widetilde{#1}}
\newcommand{\us}{\hspace{0.5pt}\rule[0pt]{4.35pt}{0.2pt}\hspace{0.2pt}}

\DeclareMathOperator{\des}{des}
\DeclareMathOperator{\row}{row}

\DeclareMathOperator{\trop}{trop}

\begin{document}
\title{Generalization of the $\epsilon$-BBS and the Schensted insertion algorithm}
\author[$\dagger$]{Katsuki Kobayashi}
\author[$\dagger$]{Satoshi Tsujimoto}
\affil[$\dagger$]{Department of Applied Mathematics and Physics, Graduate School of Informatics, Kyoto University, Kyoto, 606-8501, Japan}
\mathtoolsset{showonlyrefs=true}
\maketitle
\begin{abstract}
  The $\epsilon$-BBS is the family of solitonic cellular automata obtained via the ultradiscretization of the elementary Toda orbits,
  which is a parametrized family of integrable systems unifying the Toda equation and the relativistic Toda equation.
  In this paper, we derive the $\epsilon$-BBS with many kinds of balls and give its conserved quantities by the Schensted insertion algorithm which is introduced in combinatorics.
  To prove this, we extend birational transformations of the continuous elementary Toda orbits to the discrete hungry elementary Toda orbits.
    \\ \\
 \noindent{Keywords\/}: ultradiscrete system, Toda lattice, box-ball system, Robinson-Schensted-Knuth correspondence
\end{abstract}
\section{Introduction}
Discrete integrable systems have been actively studied from various viewpoints.
One of the most remarkable applications of discrete integrable systems is the derivation of integrable cellular automata via a limiting procedure called \textit{ultradiscretization}.
The earliest example of the integrable cellular automaton is Takahashi-Satsuma's box-ball system (BBS) \cite{TS}, which can be obtained from the ultradiscretization of the discrete Lotka-Volterra equation \cite{TTMS} and the discrete Toda equation \cite{NTT}.
Ultradiscretization also reveals an unexpected relationship between integrable systems and combinatorics.
For example, the ultradiscretization of the discrete Toda equation can be regarded as the Schensted insertion algorithm for Young tableaux if the variables are interpreted properly \cite{ET, NY}.
The BBS is also known to be closely related to the combinatorics.
In this paper, we focus on the result by Fukuda \cite{F} which shows that the P-symbol of the Robinson-Schensted-Knuth (RSK) correspondence is a conserved quantity of the BBS with many kinds of balls \cite{T} (it is called the generalized BBS in \cite{F}).

The objective of this paper is to show that the P-symbol of the RSK correspondence is a conserved quantity of a generalization of the $\epsilon$-BBS introduced in \cite{K}.
The $\epsilon$-BBS is the box-ball system obtained by ultradiscretizing the elementary Toda orbits \cite{FG1}, which is a family of integrable systems containing
the (ordinary) Toda lattice and the relativistic Toda lattice as special cases.
We consider what is known as the \textit{hungry extension} of the discrete elementary Toda orbits, by which we obtain the $\epsilon$-BBS with many kinds of balls (we call it the \textit{hungry $\epsilon$-BBS}).
Fukuda's original proof of the conservation of the P-symbol was based on the fact that the time evolution of the BBS with many kinds of balls is described by successive applications of elementary Knuth transformations.
However, this method cannot be extended to the $\epsilon$-BBS in an obvious way.
Thus, to prove the conservation of the P-symbol, we use birational transformations among the elementary Toda orbits \cite{FG1}.
It was shown in \cite{FG1} that transformations commute with the time evolution of (continuous) elementary Toda orbits.
In this paper, we show that the same transformations even commute with the time evolution of discrete hungry elementary Toda orbits, thus generalizing the result of \cite{FG1}.
Then, we show that Noumi-Yamada's geometric Schensted insertion \cite{NY} is invariant under this transformation.

The paper is organized as follows. In Section 2, we present some notations and review the basic properties of Young tableaux.
In Section 3, we derive the discrete hungry elementary Toda orbits and ultradiscretize them to obtain the hungry $\epsilon$-BBS.
In Section 4, we introduce and prove the main theorem of this paper, which states that the P-symbol of the RSK map is a conserved quantity of the hungry $\epsilon$-BBS.
Section 5 gives the conclusion.
\section{Preliminaries}
\subsection{Notations and basic properties of Young tableaux}
In this section, we provide some basic facts about Young tableaux and notations used throughout the paper.
Let $[m] = \{1, 2, ..., m\}$ be a set of $m$ letters equipped with the usual ordering on integers.
A finite sequence $v = v_1v_2 \cdots v_l$ using the letters $[m]$ is called a \textit{word}. A word $v$ is \textit{non-decreasing} if $v_k \leq v_{k+1}$ for each $1 \leq k < l$.
Let $\lambda = (\lambda_1, \lambda_2, ..., \lambda_k)$, where $\lambda_1 \geq \lambda_2 \geq \cdots \geq \lambda_k \geq 0, ~ (k \leq m)$ be a Young diagram.
A \textit{semi-standard tableau} (SST) of shape $\lambda$ is obtained by assigning a letter in $[m]$ to each box of $\lambda$ so as to satisfy the followings:
\begin{itemize}
  \item In each row, the letters are non-decreasing from left to right, and
  \item In each column, the letters are strictly increasing from top to bottom.
\end{itemize}
The \textit{Schensted insertion} of a letter $i \in [m]$ into an SST $T$ is defined as follows:
\begin{enumerate}
  \item Set $k := 1$ and $x := i$.
  \item Find the leftmost letter in the $k$-th row of $T$ that is greater than $x$. If no such letter is found, then append $x$ to the right end of the $k$-th row of $T$, and then terminate.
   If such a letter $j$ is found, then replace it with $x$ and set $x := j,~ k := k + 1$, then go back to the start of Step 2.
\end{enumerate}
The SST obtained by inserting a letter $i$ into an SST $T$ is denoted by $T \leftarrow i$.
The following is an example of the Schensted insertion:
  \begin{align}
    \begin{ytableau}
      1 & 1 & 2 & 3  \\
      2 & 3 & 3 & 4 \\
      3 & 4 & 6
    \end{ytableau} \quad
    \begin{ytableau}
      \none  \\
      \none[\longleftarrow] & \none[1] & \none[=] \\
      \none
    \end{ytableau} \quad
    \begin{ytableau}
      1 & 1 & 1 & 3  \\
      2 & 2 & 3 & 4 \\
      3 & 3 & 6 \\
      4
    \end{ytableau}.
  \end{align}
Let $w = w_1w_2\cdots w_l$ be a word. We define $T \leftarrow w$ as an SST obtained by $((((T \leftarrow x_1) \leftarrow x_2) \leftarrow \cdots ) \leftarrow x_l ) $.
Let $T$ be an SST. Denote each row of $T$ by $r_1, r_2, ..., r_k$. The \textit{row word} $w_{\row}(T)$ of $T$ is defined as $w_{\row}(T) = r_kr_{k-1}\cdots r_2r_1$.
A product of two SSTs $T$ and $T'$ is defined as $T \cdot T' := T \leftarrow w_{\row}(T')$. We will use the following property of this product later.
\begin{prop}
  The product defined above is associative, that is, for any SST $T_1, T_2$ and $T_3$, we have $(T_1 \cdot T_2) \cdot T_3 = T_1 \cdot (T_2 \cdot T_3)$.
\end{prop}
See for example \cite{Fu} for a proof of Proposition 2.1.
Two different words can give rise to the same SST by the Schensted insertion:
\begin{align}
    \varnothing \leftarrow 132 ~ =  ~
  \begin{ytableau}
    1 & 2  \\
    3
  \end{ytableau}, \quad
  \varnothing \leftarrow 312 ~ =  ~
\begin{ytableau}
  1 & 2  \\
  3
\end{ytableau}.
\end{align}
\begin{defi}
  The following two transformations (and their inverse) for three consecutive letters in a word are called \textit{elementary Knuth transformation}:
  \begin{align}
    &yzx \mapsto yxz, \quad x < y \leq z, \label{k1} \\
    &xzy \mapsto zxy, \quad x \leq y < z. \label{k2}
  \end{align}
  When two words $w$ and $w'$ are transformed into each other by a finite sequence of elementary Knuth transformations, we say that $w$ and $w'$ are \textit{Knuth equivalent}.
\end{defi}
\begin{prop}
  The two words $w$ and $w'$ are Knuth equivalent if and only if $\varnothing \leftarrow w$ and  $\varnothing \leftarrow w'$ give the same SST.
\end{prop}
Let $u_1, u_2, ..., u_l \in [m]$ and $v_1, v_2, ..., v_l \in [n]$.
An array consisting of two rows
\begin{align}
  \omega = \left(\begin{array}{cccc}
   u_1&u_2&\cdots&u_l\\
   v_1&v_2&\cdots&v_l
  \end{array}\right)
\end{align}
is called a \textit{biword} if the following conditions are satisfied:
\begin{itemize}
  \item $u_1 \leq u_2 \leq \cdots \leq u_l$, and
  \item For all $1 \leq i < l$, $v_i \leq v_{i+1}$ if $u_i = u_{i+1}$.
\end{itemize}
We write the first row of $\omega$ as $u$ and the second row of $\omega$ as $v$.
The \textit{P-symbol} $P(\omega)$ is the SST obtained by $\varnothing \leftarrow v$.
The \textit{Q-symbol} $Q(\omega)$ is the SST of the same shape as $P(\omega)$ constructed as follows:
The tableau $Q(\omega)$ is obtained by adding the box with letter $u_k$ to the place where $v_k$ is inserted in $P(\omega)$.
The correspondence $\omega \mapsto (P(\omega), Q(\omega))$ is called the \textit{RSK correspondence}.
\begin{prop}
  The RSK correspondence
  \begin{align}
    \omega \mapsto (P(\omega), Q(\omega))
  \end{align}
  gives a bijection between biwords and tuples of SSTs of the same shape.
\end{prop}
In this paper we consider only the P-symbol.
\subsection{Piecewise-linear formula for Schensted insertion}
The Schensted insertion can be written in the form of a piecewise-linear equation.
It first appeared in \cite{Ki} and further investigated in \cite{NY} where its relation to the discrete Toda lattice was identified.
Let $v$ and $w$ be non-decreasing words consisting of letters in $[m]$.
Let $x = (x_1, x_2, ..., x_m)$ and $a = (a_1, a_2, ..., a_m)$ be \textit{coordinate representations} of $v$ and $w$, respectively, i.e.,
$x_i$ (resp. $a_i$) is the number of $i$'s in the word $v$ (resp. $w$).
The SST obtained by a Schensted insertion $w \leftarrow v$ consists of two rows, with its first row denoted by $w'$ and the second row by $v'$.
Let $y = (y_1, y_2, ..., y_m)$ and $b = (b_1, b_2, ..., b_m)$ be coordinate representations of $w'$ and $v'$, respectively.
There exists a piecewise-linear formula to compute $y$ and $b$ from $x$ and $a$. First, we define $\eta_i, ~i = 1,2, ..., m$, as
\begin{align}
  \eta_1 = y_1, \quad \eta_j = \eta_{j-1} + y_j, ~j = 2, 3, ..., m.
\end{align}
Then, $\eta_j, ~j = 1, 2, ..., m$, is expressed in terms of $x$ and $a$ as
\begin{align}
  \eta_j = \max_{1 \leq k \leq j}\{x_1 + x_2 + \cdots + x_k + a_k + a_{k+1} + \cdots + a_j \}, \label{Si}
\end{align}
by which we can recover $y_i$'s and $x_i$'s because $x_i + a_i = y_i + b_i$ holds for all $i = 1, 2, ... ,m$.
The proof of formula \eqref{Si} is given in \cite{NY}.
\subsection{Box-ball system}
In this section, we review the work by Fukuda \cite{F}, in which the P-symbol of the RSK correspondence was shown to be a conserved quantity of the generalized BBS.
First, we present the definition of the generalized BBS. Let $u = (u_i)_{i=0}^{\infty}$ be a semi-infinite sequence of letters in $[m] \cup \{e\}$ and $u_i = e$ for all but finitely many $i \in \mathbb{Z}_{\geq 0}$.
We regard the letter $e$ to be greater than any element of $[m]$.
Let $\Omega$ be the set of all such sequences.
The letter $i \in [m]$ represents a `ball of color $i$' and $e$ represents an `empty box'.
We define the map $T \colon \Omega \to \Omega$ as follows:
\begin{enumerate}
  \item Set $i := 1$.
  \item Move the leftmost ball of color $i$ to the nearest empty box on the right. Repeat this procedure for the other balls of color $i$ until all of them have been moved once.
  \item If $i = m$, then terminate. Otherwise, set $i := i+1$ and go to Step $2$.
\end{enumerate}
For an initial sequence $u^{(0)} \in \Omega$, the time evolution of the generalized BBS is defined as $u^{(t+1)} = T(u^{(t)})$.
The figure below shows an example of the time evolutions of the generalized BBS (here, the letter $e$ is replaced by an underscore symbol `$\us$'. ).
\begin{align}
  t = 0:~&\us132\us\us12\us413\us\us\us\us\us\us\us\us\us\us\us\us\us\us\us\us \\
  t = 1:~&\us\us\us312\us\us1\us2413\us\us\us\us\us\us\us\us\us\us\us\us\us\us \\
  t = 2:~&\us\us\us\us3\us12\us1\us\us2413\us\us\us\us\us\us\us\us\us\us\us\us \\
  t = 3:~&\us\us\us\us\us3\us\us121\us\us\us2413\us\us\us\us\us\us\us\us\us\us \\
  t = 4:~&\us\us\us\us\us\us3\us\us\us211\us\us\us2413\us\us\us\us\us\us\us\us \\
  t = 5:~&\us\us\us\us\us\us\us3\us\us\us2\us11\us\us\us2413\us\us\us\us\us\us \\
  t = 6:~&\us\us\us\us\us\us\us\us3\us\us\us2\us\us11\us\us\us2413\us\us\us\us
\end{align}
For $u \in \Omega$, let $f(u)$ denote a finite subsequence of $u$ obtained by removing all $e$'s.
For $u^{(0)}$ in the above example, we have $f(u^{(0)}) = 13212413$.
\begin{prop}[\cite{F}]
  For any $u \in \Omega$, the following two SST coincide:
  \begin{align}
    \varnothing \leftarrow f(u), \quad \varnothing \leftarrow f(T(u)).
  \end{align}
  That is, the P-symbol of the RSK correspondence gives a conserved quantity of the generalized BBS.
\end{prop}
Although the proof of Proposition 2.4 can be found in \cite{F}, to make this paper self-contained, we write the proof here. The idea is to realize the time evolution of the generalized BBS by successive applications of Knuth transformations.
First, let us rewrite the above time evolution rule into a \textit{carrier rule}.
We use above sequences $u^{(0)} = e132ee12e413ee\cdots$ and $u^{(1)} = T(u^{(0)}) = eee312ee1e2413ee\cdots$ as examples. First, let $N$ be the number of indices $i$ such that $u_i \neq 0$.
Let $C^{(0)} = \underbrace{ee\cdots}_{N}$ be a finite sequence consisting of $N$ copies of $e$'s, which is called a \textit{carrier}. We consider a sequence $v^{(0)} = (v_i^{(0)})_{i=0}^{\infty}$ obtained by concatenating the carrier $C^{(0)}$ to the left end of
$u^{(0)}$:
\begin{align}
  v^{(0)} &= \underline{eeeeeeee}e132ee12e413ee\cdots
\end{align}
Here, the part of the sequence $v^{(0)}$ corresponding to the carrier is underlined.
We also use the notation $v^{(0)} = C^{(0)}e132ee12e413ee\cdots$ for the same sequence.
We define $C^{(1)}$ and the sequence $v^{(1)} = v_0^{(1)}C^{(1)}v_{N+1}^{(1)}v_{N+2}^{(1)}\cdots$ from $v^{(0)}$ as follows:
Let $x$ be the letter $v^{(0)}_{N}$. In other words, $x$ is the letter to the right of the carrier $C^{(0)}$.
Then, we perform procedure (A) or (B) depending on whether there is a letter in $C^{(0)}$ greater than $x$.
\begin{enumerate}
  \item[(A)] If there is a letter in $C^{(0)}$ that is greater than $x$, we name the leftmost one $y$.
  Then, we
  \begin{enumerate}
    \item[(i)] replace $y$ with $x$, then
    \item[(ii)]  remove $y$ from the carrier and concatenate it to the left of the carrier.
  \end{enumerate}
  \item[(B)] If there is no letter in $C^{(0)}$ that is greater than $x$, we first append $x$ to the rightmost position of the carrier.
  Then, we remove the leftmost letter $y$ of the carrier and concatenate it to the left of the carrier.
\end{enumerate}
This results in a new carrier $C^{(1)}$ and a sequence $v^{(1)} = v_0^{(1)}C^{(1)}v_{N+1}^{(1)}v_{N+2}^{(1)}\cdots$.
We continue the above procedure until there are no balls to the right of the carrier and the carrier consists of $N$ copies of $e$.
An example is given below.
\begin{align}
  v^{(0)} = &~\underline{eeeeeeee}~e132ee12e413eeeeeee\cdots \quad& v^{(8)} &= eee312ee~\underline{12eeeeee}~e413eeeeeee\cdots \\
  v^{(1)} = &e~\underline{eeeeeeee}~132ee12e413eeeeeee\cdots \quad& v^{(9)} &= eee312ee1~\underline{2eeeeeee}~413eeeeeee\cdots\\
  v^{(2)} = &ee~\underline{1eeeeeee}~32ee12e413eeeeeee\cdots \quad& v^{(10)} &= eee312ee1e~\underline{24eeeeee}~13eeeeeee\cdots\\
  v^{(3)} = &eee~\underline{13eeeeee}~2ee12e413eeeeeee\cdots \quad& v^{(11)} &= eee312ee1e2~\underline{14eeeeee}~3eeeeeee\cdots\\
  v^{(4)} = &eee3~\underline{12eeeeee}~ee12e413eeeeeee\cdots \quad& v^{(12)} &= eee312ee1e24~\underline{13eeeeee}~eeeeeee\cdots\\
  v^{(5)} = &eee31~\underline{2eeeeeee}~e12e413eeeeeee\cdots \quad& v^{(13)} &= eee312ee1e241~\underline{3eeeeeee}~eeeeee\cdots \\
  v^{(6)} = &eee312~\underline{eeeeeeee}~12e413eeeeeee\cdots \quad& v^{(14)} &= eee312ee1e2413~\underline{eeeeeeee}~eeeee\cdots\\
  v^{(7)} = &eee312e~\underline{1eeeeeee}~2e413eeeeeee\cdots
\end{align}
Finally, we delete the carrier $\underline{eeeeeeee}$ from $v^{(0)}$ and $v^{(14)}$ to obtain $u^{(0)} \mapsto u^{(1)}$.
The equivalence between this procedure and the time evolution rule defined above is shown in \cite{F}.
Procedures (A) and (B) above can be realized by a sequence of elementary Knuth transformations. For (B) the assertion is obvious as it does not change sequence itself.
Let us consider the case of (A). Let  $x_1x_2\cdots x_ly z_1 z_2...z_{k-1}z_k$ be the state of the carrier.
Step (i) of the procedure (A) can be achieved via a sequence of elementary Knuth transformations \eqref{k1} ($bca \mapsto bac$ for $a < b \leq c$) as follows:
\begin{align}
  &x_1~x_2~\cdots~x_{l-1}~x_l~y~z_1~z_2~\cdots~z_{k-2}~z_{k-1}~\overline{z_k~x} \\
  &x_1~x_2~\cdots~x_{l-1}~x_l~y~z_1~z_2~\cdots~z_{k-2}~\overline{z_{k-1}~x}~z_k \\
  & \quad \quad \vdots \\
  &x_1~x_2~\cdots~x_{l-1}~x_l~y~\overline{z_1~x}~z_2~\cdots~z_{k-2}~z_{k-1}~z_k \\
  &x_1~x_2~\cdots~x_{l-1}~x_l~y~x~z_1~z_2~\cdots~z_{k-2}~z_{k-1}~z_k
\end{align}
Here, the letters whose positions are to be exchanged are marked as $\overline{a ~ b}$.
Next, we evacuate $x$ from the carrier (Step (ii) of the procedure (A)) by a sequence of elementary Knuth transformations \eqref{k2} ($acb \mapsto cab$ for $a \leq b < c$) as follows:
\begin{align}
  &x_1~x_2~\cdots~x_{l-1}~\overline{x_l~y}~x~z_1~z_2~\cdots~z_{k-1}~z_k \\
  &x_1~x_2~\cdots~\overline{x_{l-1}~y}~x_l~x~z_1~z_2~\cdots~z_{k-1}~z_k \\
  & \quad \quad \vdots \\
  &\overline{x_1~y}~x_2~\cdots~x_{l-1}~x_l~x~z_1~z_2~\cdots~z_{k-1}~z_k \\
  &y~x_1~x_2~\cdots~x_{l-1}~x_l~x~z_1~z_2~\cdots~z_{k-1}~z_k
\end{align}
Thus, two words $w = ``eeeeeeeee132ee12e413''$ and $w' = ``eee312ee1e2413eeeeeeee''$ result in the same SST according to Proposition 2.2.
Since removing the letter $e$ from a word does not affect the position of letters in $[m]$ after the Schensted insertion, we have $\varnothing \leftarrow f(w) = \varnothing \leftarrow f(w')$. This concludes the proof.
\section{Hungry $\epsilon$-BBS}
In \cite{K}, a family of box-ball systems called the $\epsilon$-BBS was introduced.
The $\epsilon$-BBS contains Takahashi-Satsuma's BBS \cite{TS} as a special case.
In this section, we first derive the discrete hungry elementary Toda orbits (d-heToda orbits) and then obtain the hungry $\epsilon$-BBS by ultradiscretizing them.
The d-heToda orbits contains a positive integer parameter $M$. When $M$ is set to $1$, then the d-heToda orbits specializes to the discrete elementary Toda orbits.
Thus, the hungry $\epsilon$-BBS is a multi-color extension of the $\epsilon$-BBS in the sense that the parameter $M$ corresponds to the number of colors of balls in the hungry $\epsilon$-BBS.
Furthermore, we present a birational transformations between different orbits of the discrete hungry elementary Toda orbits.
\subsection{Discrete hungry elementary Toda orbits}
Let $N$ be a positive integer and  $\epsilon = (\epsilon_0, \epsilon_1, ..., \epsilon_{N-1}) \in \{0, 1\}^{N}$.
We define $R^{(t)}, L_1^{(t)}$ and $L_2^{(t)}$ as
\begin{align}
  &R^{(t)} := \sum_{i=1}^{N} q^{(t)}_{i-1}E_{i, i} + \sum_{i=1}^{N-1} E_{i, i+1}, \\
  &(L_1^{(t)})^{-1} = I_N + \sum_{i=1}^{N-1} -\epsilon_{i-1} e_{i-1}^{(t)} E_{i+1, i}, \\
  &L_2^{(t)} = I_N + \sum_{i=1}^{N-1} (1-\epsilon_{i-1}) e_{i-1}^{(t)} E_{i+1, i}.
\end{align}
Let $M$ be an integer greater than or equal to $1$. We call the following equation the \textit{discrete hungry elementary Toda orbits} (d-heToda orbits):
\begin{align}
  L_1^{(t+1)}L_2^{(t+1)}R^{(t+M)} = R^{(t)}L_1^{(t)}L_2^{(t)}. \label{Lax}
\end{align}
Equation \eqref{Lax} is equivalent to the following system of equations:
\begin{align}
 &q_i^{(t+M)} = q_i^{(t)} + e_i^{(t)} - e_{i-1}^{(t+1)}, \quad i = 0, 1, ..., N-1 \label{ht1} \\
 &e_i^{(t+1)} = \cfrac{q_{i+1}^{(t)} + \epsilon_{i+1} e_{i+1}^{(t)}}{q_i^{(t+M)} + \epsilon_{i}e_{i-1}^{(t+1)}}~e_i^{(t)}, \quad i = 0, 1, ..., N-2, \label{ht2}
\end{align}
where $e_{-1}^{(t)} = 0$ for all $t$.
We regard the system of equations \eqref{ht1}, \eqref{ht2} as a time evolution
\begin{align}
  (q_i^{(t)}, q_i^{(t+1)}, ..., q_i^{(t+M-1)})_{i=0}^{N-1}, (e_i^{(t)})_{i=0}^{N-2}  \mapsto (q_i^{(t+M)}, q_i^{(t+M+1)}, ..., q_i^{(t+2M-1)})_{i=0}^{N-1}, (e_i^{(t+M)})_{i=0}^{N-2}. \\ \label{time}
\end{align}
We call an orbit of \eqref{time} through any given initial value an \textit{$\epsilon$-orbit}.

Let $X^{(t)} = L_1^{(t)} L_2^{(t)}R^{(t+M-1)}R^{(t+M-2)}\cdots R^{(t)}$. Then we have
\begin{align}
  X^{(t+M)} = (L_1^{(t)}L_2^{(t)})^{-1} X^{(t)} L_1^{(t)}L_2^{(t)}.
\end{align}
Thus the characteristic polynomial $\phi(x) = \det(X^{(t)} - x I_{N})$ is a conserved quantity of the d-heToda orbits.
The system \eqref{ht1} and \eqref{ht2} can be rewritten in the following \textit{subtraction-free} form:
\mathtoolsset{showonlyrefs=false}
\begin{align}
  &\begin{cases}
    d_i^{(t+1)} = \cfrac{q_i^{(t)}}{q_{i-1}^{(t+M)}} d_{i-1}^{(t+1)}\\
    q_i^{(t+M)} = d_{i}^{(t+1)} + e_i^{(t)}
  \end{cases}, \quad (\epsilon_{i-1}, \epsilon_{i} ) = (0, 0) \label{eq1} \\
  &\begin{cases}
    d_i^{(t+1)} =  q_{i}^{(t)} + e_i^{(t)}\\
    q_i^{(t+M)} = \cfrac{d_i^{(t+1)}}{d_{i-1}^{(t+1)}} q_{i-1}^{(t)}
  \end{cases}, \quad (\epsilon_{i-1}, \epsilon_{i} ) = (1, 1) \label{eq2} \\
  &\begin{cases}
    d_i^{(t+1)} =  q_{i}^{(t)} + e_i^{(t)}\\
    q_i^{(t+M)} = \cfrac{d_i^{(t+1)}}{q_{i-1}^{(t+M)}} d_{i-1}^{(t+1)}
  \end{cases}, \quad (\epsilon_{i-1}, \epsilon_{i} ) = (0, 1) \label{eq3} \\
  &\begin{cases}
    d_i^{(t+1)} = \cfrac{q_i^{(t)}}{d_{i-1}^{(t+1)}} q_{i-1}^{(t)}\\
    q_i^{(t+M)} = d_{i}^{(t+1)} + e_i^{(t)}
  \end{cases}, \quad (\epsilon_{i-1}, \epsilon_{i} ) = (1, 0). \label{eq4} \\
  &e_i^{(t+1)} = \cfrac{q_{i+1}^{(t)} + \epsilon_{i+1} e_{i+1}^{(t)}}{q_i^{(t+M)} + \epsilon_{i}e_{i-1}^{(t+1)}}~e_i^{(t)}, \label{eq5}
\end{align}
\mathtoolsset{showonlyrefs=true}
where $q_{-1}^{(t)}, d_{-1}^{(t)} \equiv 1$ and $e_{N-1}^{(t)} \equiv 0$.
We note that, since the right-hand sides of equations \eqref{eq1}--\eqref{eq5} do not contain any subtraction, they can be ultradiscretized.
The BBS obtained from the ultradiscretization of \eqref{eq1}--\eqref{eq5} will be discussed later.
\subsection{Birational transformations of discrete hungry elementary Toda orbits}
There is a birational transformation from the $\epsilon$-orbit for given $\epsilon \in \{0, 1\}^N$ to the $\epsilon'$-orbit for another parameter $\epsilon' \in \{0, 1\}^N$.
Such a birational transformation is considered in \cite{FG1} for the continuous case. In this section, we extend it to the discrete and hungry case.
We define $E_i^{(t)} = I_N + e_{i}^{(t)}E_{i+2,i+1}$ for $i=0, 1, ..., N-2$.
Let $I = \{i_0 < i_1 < \cdots < i_{k-1} \mid \epsilon_{i_j+1} = 0\}$.
We denote $E_{i_l}^{(t)}E_{i_{l}-1}^{(t)}\cdots E_{i_{l-1}+1}^{(t)}$ by $E_{[i_l, i_{l-1}+1]}^{(t)}$ where $i_{-1} = -1$.
Then,
\begin{align}
  L_1^{(t)} L_2^{(t)} = E_{[i_0, 0]}^{(t)}E_{[i_1, i_0+1]}^{(t)} \cdots E_{[i_{k-1}, i_{k-2}+1]}^{(t)} . \label{LE}
\end{align}
\begin{ex}
  When $N = 6, \epsilon = (0,1,1,0,1,0)$, we have $i_0 = 2, i_1 = 4$ and
  \begin{align}
    L_1^{(t)}L_2^{(t)} = E_2^{(t)} E_1^{(t)} E_0^{(t)} E_4^{(t)} E_3^{(t)}.
  \end{align}
\end{ex}
Suppose there is an index $i$ such that $\epsilon_i = 0, \epsilon_{i+1} = 1$.
Define $\epsilon' = (\epsilon'_0, \epsilon'_1, ..., \epsilon'_{N-1})$ as
\begin{align}
  \epsilon'_j = \begin{cases}
      1 & j = i, \\
      0 & j = i + 1,\\
      \epsilon_j & {\rm otherwise}
  \end{cases}
\end{align}
Then consider the following transformations of matrices:
\begin{align}
 E_{i}^{(t)} R^{(t+M-1)}R^{(t+M-2)}\cdots R^{(t)} &= \wt{R}^{(t+M-1)}E_i^{(t, 1)}R^{(t+M-2)} \cdots R^{(t)}, \\
                                                  &= \wt{R}^{(t+M-1)}\wt{R}^{(t+M-2)}E_i^{(t, 2)} \cdots R^{(t)}, \\
                                                  &\vdots \\
                                                  &= \wt{R}^{(t+M-1)}\wt{R}^{(t+M-2)} \cdots \wt{R}^{(t)}E_i^{(t, M)}.
\end{align}
where $\wt{R}^{(t+M-j)}$ and $E_i^{(t, j)}$ are matrices of the form
\begin{align}
  &\wt{R}^{(t+M-j)} = \sum_{l=1}^{N} \wt{q}^{(t+M-j)}_{l-1}E_{l, l} + \sum_{l=1}^{N-1} E_{l, l+1}, \\
  &E_i^{(t, j)} = I_N + e_{i}^{(t, j)}E_{i+2,i+1}, \quad j = 1, 2, ..., M.
\end{align}
We denote $\wt{E}_j^{(t)} = E_j^{(t, M)}$ and $\wt{e}_j^{(t)} = e_j^{(t, M)}$.
Let us express $\wt{q}_i^{(t+M-j)}, \wt{q}_{i+1}^{(t+M-j)}$ and $\wt{e}_i^{(t)}$ by $q_i^{(t+M-j)}, q_{i+1}^{(t+M-j)}$ and $e_i^{(t)}$.
Let $e_i^{(t, 0)} := e_i^{(t)}$. Then
\mathtoolsset{showonlyrefs=false}
\begin{align}
  &e_i^{(t, j)} = \cfrac{e_{i}^{(t, j-1)}q_i^{(t+M-j)}}{e_i^{(t, j-1)} + q_{i+1}^{(t+M-j)}}, \label{B1}\\
  &\wt{q}_i^{(t+M-j)} = \cfrac{q_{i+1}^{(t+M-j)}q_i^{(t+M-j)}}{e_i^{(t, j-1)} + q_{i+1}^{(t+M-j)}}, \label{B2} \\
  &\wt{q}_{i+1}^{(t+M-j)} = e_i^{(t, j-1)} + q_{i+1}^{(t+M-j)}, \label{B3} \\
  &\wt{q}_{l}^{(t)} = q_l^{(t)}, \quad l \neq i, i+1, \label{B4} \\
  &\wt{e}_l^{(t)} = e_l^{(t)}, \quad l \neq i. \label{B5}
\end{align}
We denote the rational transformation \eqref{B1}--\eqref{B5} by $\varphi_i,~i=0, 1, ..., N-2$.
\mathtoolsset{showonlyrefs=true}
\begin{prop}
  The transformation $\varphi_i$ commutes with the time evolution of the d-hetoda orbits.
\end{prop}
\begin{proof}
  Suppose $\epsilon_0 = 0$ and $I = \{i_0 < i_1 < \cdots < i_{k-1} \mid \epsilon_{i_j+1} = 0\}$. Then $L_1^{(t)}L_2^{(t)}$ has the form
  \begin{align}
    L_1^{(t)} L_2^{(t)} = E_{[i_0, 0]} E_{[i_1, i_0+1]} \cdots E_{[i_{k-1}, i_{k-2}+1]}.
  \end{align}
  We consider the case $i = 0$ and $\epsilon_{1} = 1$. In this case, $E_{[i_0, 0]}$ is the product of two or more matrices since $i_0 \geq 1$. The general case can be shown in the same way.
  From the definition of the $\epsilon$-orbits \eqref{Lax}, we have
  \begin{align}
    R^{(t+l)}E_{[i_0, 0]}^{(t+l)} E_{[i_1, i_0+1]}^{(t+l)} \cdots E_{[i_{k-1}, i_{k-2}+1]}^{(t+l)} = E_{[i_0, 0]}^{(t+l+1)} E_{[i_1, i_0+1]}^{(t+l+1)} \cdots E_{[i_{k-1}, i_{k-2}+1]}^{(t+l+1)} R^{(t+l+M)}, \label{time2}
  \end{align}
  for $l = 0, 1, ..., M-1$. From the definition of the birational transformation \eqref{B1}--\eqref{B5}, we also have
  \begin{align}
    E_0^{(t, l)}R^{(t+M-1-l)} = \wt{R}^{(t+M-1-l)}E_0^{(t, l+1)}, \quad l = 0, 1, ..., M-1, \label{tr}
  \end{align}
  where $E_0^{(t, 0)} := E_0^{(t)}$.
  We define $\wt{E}_0^{(t)} := E_0^{(t, M)}$ and $\wt{E}_l^{(t)} := E_l^{(t)}$ for $l = 1, 2, ..., N-2$.
  The time evolution of the $\epsilon'$-orbits for $\epsilon' = (1, 0, \epsilon_2, \epsilon_3, ..., \epsilon_{N-1})$ is
  \begin{align}
    \wt{R}^{(t+l)}\wt{E}_{0}^{(t+l)}\wt{E}_{[i_0, 1]}^{(t+l)} \wt{E}_{[i_1, i_0+1]}^{(t+l)} \cdots & \wt{E}_{[i_{k-1}, i_{k-2}+1]}^{(t+l)}  \\
    &=\wt{E}_{0}^{(t+l+1)} \wt{E}_{[i_0, 1]}^{(t+l+1)} \wt{E}_{[i_1, i_0+1]}^{(t+l+1)} \cdots \wt{E}_{[i_{k-1}, i_{k-2}+1]}^{(t+l+1)} \wt{R}^{(t+l+M)}, \label{trtime}
  \end{align}
  for $l = 0, 1, ..., M-1$. We define matrices $\overline{R}^{(t+M+l)}$ for $l = 0,1, ..., M-1$ and $\overline{E}_0^{(t+M)}$ by
  \begin{align}
    E_0^{(t+M, l)}R^{(t+2M-1-l)} = \overline{R}^{(t+2M-1-l)}E_0^{(t+M, l+1)}, \quad l = 0, 1, ..., M-1 \\ \label{timetr}
  \end{align}
  where $E_0^{(t+M, 0)} := E_0^{(t+M)}$ and $\overline{E}_0^{(t+M)} = E_0^{(t+M, M)}$. We also define $\overline{E}_{l}^{(t+M)} = E_{l}^{(t+M)}$ for $l = 1, 2, ..., N-2$.  We must show that $\wt{R}^{(t+M+l)} = \overline{R}^{(t+M+l)}$ for $l = 0, 1, ..., M-1$ and $\wt{E}_l^{(t+M)} = \overline{E}_l^{(t+M)}$ for $l = 0, 1, ..., N-2$.
  From \eqref{tr} and \eqref{trtime}, we have
  \begin{align}
    \wt{R}^{(t)}\wt{E}_{0}^{(t)}E_{[i_0, 1]}^{(t)} E_{[i_1, i_0+1]}^{(t)} \cdots E_{[i_{k-1}, i_{k-2}+1]}^{(t)} &= \wt{R}^{(t)}\wt{E}_{0}^{(t)}\wt{E}_{[i_0, 1]}^{(t)} \wt{E}_{[i_1, i_0+1]}^{(t)} \cdots \wt{E}_{[i_{k-1}, i_{k-2}+1]}^{(t)} \\
    &= E_0^{(t, M-1)} R^{(t)} \wt{E}_{[i_0, 1]}^{(t)} \wt{E}_{[i_1, i_0+1]}^{(t)} \cdots \wt{E}_{[i_{k-1}, i_{k-2}+1]}^{(t)} \\
    &= E_0^{(t, M-1)}\wt{E}_{[i_0, 1]}^{(t+1)} \wt{E}_{[i_1, i_0+1]}^{(t+1)} \cdots \wt{E}_{[i_{k-1}, i_{k-2}+1]}^{(t+1)} \wt{R}^{(t+M)}. \\\label{aaa}
  \end{align}
  Comparing \eqref{aaa} with \eqref{trtime}, we obtain $\wt{E}_0^{(t+1)} = E_0^{(t, M-1)}$.
  From \eqref{time2}, we have
  \begin{align}
    R^{(t)}E_{[i_0, 0]}^{(t)} E_{[i_1, i_0+1]}^{(t)} \cdots E_{[i_{k-1}, i_{k-2}+1]}^{(t)} &= R^{(t)}E_{[i_0, 1]}^{(t)} E_{[i_1, i_0+1]}^{(t)} \cdots E_{[i_{k-1}, i_{k-2}+1]}^{(t)} E_0^{(t)} \\
    &= E_{[i_0, 1]}^{(t+1)} E_{[i_1, i_0+1]}^{(t+1)} \cdots E_{[i_{k-1}, i_{k-2}+1]}^{(t+1)} R'^{(t+M)} E_0^{(t)}, \\
    &= E_{[i_0, 1]}^{(t+1)} E_{[i_1, i_0+1]}^{(t+1)} \cdots E_{[i_{k-1}, i_{k-2}+1]}^{(t+1)} E_0^{(t+1)} R^{(t+M)},
  \end{align}
  for an upper bidiagonal matrix $R'^{(t+M)}$, owing to the relation $E_{\alpha}^{(t)}E_{\beta}^{(t)} = E_{\beta}^{(t)}E_{\alpha}^{(t)}$ for $|\alpha - \beta| > 1$ and $i_0 \geq 1$. Thus we obtain
  \begin{align}
    &\wt{E}_l^{(t+1)} = E_l^{(t+1)}, \quad l = 1, 2, ..., N-2, \\
    &R'^{(t+M)} = \wt{R}^{(t+M)}, \\
    &\wt{R}^{(t+M)} E_0^{(t)} = E_0^{(t+1)}R^{(t+M)},
  \end{align}
  owing to the uniqueness of the LU-decomposition.
  By repeating this argument inductively, we obtain
  \begin{align}
    &\wt{E}_0^{(t+k+1)} = E_0^{(t, M-k-1)}, \label{bbb1} \\
    &\wt{E}_l^{(t+k+1)} = E_l^{(t+k+1)}, \quad l = 1, 2, ..., N-2, \\
    &\wt{R}^{(t+M+k)} E_0^{(t+k)} = E_0^{(t+k+1)}R^{(t+M+k)}, \label{ccc}
  \end{align}
  for $k = 0, 1, ..., M-1$. From \eqref{timetr} and \eqref{ccc}, we have
  \begin{align}
    \overline{R}^{(t+2M-1)}E_0^{(t+M, 1)} = \wt{R}^{(t+2M-1)}E_0^{(t+M-1)}.
  \end{align}
  Thus, we have $\overline{R}^{(t+2M-1)} = \wt{R}^{(t+2M-1)}$ and $E_0^{(t+M, 1)} = E_0^{(t+M-1)}$.
  By repeating this argument inductively, we obtain
  \begin{align}
    &\overline{R}^{(t+2M-l)} = \wt{R}^{(t+2M-l)}, \\
     &E_0^{(t+M, l)} = E_0^{(t+M-l)}, \label{bbb2}
  \end{align}
  for $l = 1, 2, ..., M$. From \eqref{bbb1} and \eqref{bbb2}, we have
  \begin{align}
    \overline{E}_0^{(t+M)} = E_0^{(t+M, M)} = E_0^{(t)} = E_0^{(t, 0)} = \wt{E}_0^{(t+M)}.
  \end{align}
  This concludes the proof.
\end{proof}
The inverse of the transformation \eqref{B1}--\eqref{B5} is also rational:
\mathtoolsset{showonlyrefs=false}
\begin{align}
  &e_i^{(t,j-1)} = \cfrac{ e_i^{(t,j)} \wt{q}_{i+1}^{(t+M-j)} }{e_i^{(t,j)} + \wt{q}_i^{(t+M-j)}}, \label{Bi1} \\
  &q_i^{(t+M-j)} = e_i^{(t,j)} + \wt{q}_i^{(t+M-j)}, \label{Bi2}\\
  &q_{i+1}^{(t+M-j)} = \cfrac{ \wt{q}_{i}^{(t+M-j)}\wt{q}_{i+1}^{(t+M-j)} }{e_i^{(t,j)} + \wt{q}_i^{(t+M-j)}}. \label{Bi3}
\end{align}
\mathtoolsset{showonlyrefs=true}
Note that the right-hand sides of the rational transformation \eqref{B1}--\eqref{B5} and \eqref{Bi1}--\eqref{Bi3} have no subtractions; this is important in the proof of the main result.
We also remark that successive application of the transformations \eqref{Bi1}--\eqref{Bi3} yields a time evolution of the discrete hungry elementary Toda orbits \eqref{Lax}.
\subsection{Hungry $\epsilon$-BBS}
Let $\varepsilon > 0$.
We consider the transformations of variables $q_i^{(t)} = e^{-Q_i^{(t)}/\varepsilon}, e_i^{(t)} = e^{-\wt{E}_i^{(t)}/\varepsilon}, d_i^{(t)} = e^{-D_i^{(t)}/\varepsilon}$.
By applying them to \eqref{eq1}--\eqref{eq5} and using
\begin{align}
  \lim_{\varepsilon \to +0} -\varepsilon \log(e^{-A/\varepsilon} + e^{-B/\varepsilon}) = \min(A, B),
\end{align}
we obtain the following piecewise-linear system:
\mathtoolsset{showonlyrefs=false}
\begin{align}
  &\begin{cases}
    D_i^{(t+1)} = Q_i^{(t)} + D_{i-1}^{(t+1)} - Q_{i-1}^{(t+M)}\\
    Q_i^{(t+M)} = \min(D_{i}^{(t+1)}, E_i^{(t)})
  \end{cases}, \quad (\epsilon_{i-1}, \epsilon_{i} ) = (0, 0) \label{ueq1} \\
  &\begin{cases}
    D_i^{(t+1)} =  \min(Q_{i}^{(t)}, E_i^{(t)})\\
    Q_i^{(t+M)} = D_i^{(t+1)} + Q_{i-1}^{(t)} - D_{i-1}^{(t+1)}
  \end{cases}, \quad (\epsilon_{i-1}, \epsilon_{i} ) = (1, 1) \label{ueq2} \\
  &\begin{cases}
    D_i^{(t+1)} =  \min(Q_{i}^{(t)}, E_i^{(t)})\\
    Q_i^{(t+M)} = D_i^{(t+1)} + D_{i-1}^{(t+1)} - Q_{i-1}^{(t+M)}
  \end{cases}, \quad (\epsilon_{i-1}, \epsilon_{i} ) = (0, 1) \label{ueq3} \\
  &\begin{cases}
    D_i^{(t+1)} = Q_i^{(t)} + Q_{i-1}^{(t)} - D_{i-1}^{(t+1)}\\
    Q_i^{(t+M)} = \min(D_{i}^{(t+1)}, E_i^{(t)})
  \end{cases}, \quad (\epsilon_{i-1}, \epsilon_{i} ) = (1, 0). \label{ueq4} \\
  &E_i^{(t+1)} = \min(Q_{i+1}^{(t)}, \mathcal{E}_{i+1} + E_{i+1}^{(t)}) - \min(Q_i^{(t+M)}, \mathcal{E}_{i} + E_{i-1}^{(t+1)}) + E_i^{(t)}, \label{ueq5}
\end{align}
\mathtoolsset{showonlyrefs=true}
for $i = 0, 1, ..., N-1$. Here, $\mathcal{E}_{i}$ is defined as
\begin{align}
  \mathcal{E}_{i} =
  \begin{cases}
    +\infty   & \epsilon_i = 0,\\
    0 & \epsilon_i = 1\\
  \end{cases}
\end{align}
and $D_{-1}^{(t+1)} \equiv 0$ and $E_{N-1}^{(t+1)} \equiv + \infty$.
We call the system \eqref{ueq1}--\eqref{ueq5} the \textit{ultradiscrete hungry elementary Toda orbits} (u-hToda).
We consider an auxiliary variable $E_{-1}^{(t)}$ and its time evolution
\begin{align}
  E_{-1}^{(t+1)} = \min(Q_0^{(t)}, \mathcal{E}_0 + E_0^{(t)}) + E_{-1}^{(t)}.
\end{align}
From $Q = (Q_i^{(0)}, Q_i^{(1)}, ..., Q_i^{(M-1)})_{i=0}^{N-1} \in \mathbb{Z}^{MN}_{> 0}$ and $E = (E_i^{(0)})_{i=-1}^{N-2} \in \mathbb{Z}^{N}_{\geq 0}$, we construct a sequence $u := \Phi_N(Q, E) \in \Omega$ by the following rule:
\begin{itemize}
  \item $Q_i^{(j)}$ denotes the number of balls of color $j+1$ in the $(i+1)$-th block of balls (the balls in each block are arranged in increasing order), and
  \item $E_i^{(0)}$ denotes the number of empty boxes between the $(i+1)$-st and the $(i+2)$-nd blocks of balls.
\end{itemize}
The map $\Phi_N$ is a bijection between $\mathbb{Z}_{> 0}^{NM} \times \mathbb{Z}_{\geq 0}^{N}$ and the set $\Omega_N (\subset \Omega)$  of sequences satisfying
\begin{align}
  \begin{cases}
    \des(u) = N-1 & u_0 \neq e \\
    \des(u) = N & u_0 = e
  \end{cases}
\end{align}
where $\des(u)$ is the number of descents in $u$, i.e., the number of indexes $i$ such that $u_i > u_{i+1}$.
The time evolution of the u-hToda orbits
\begin{align}
  &Q_i^{(j)} \mapsto Q_i^{(j+M)}, \quad  0 \leq i \leq N-1, ~ 0 \leq j \leq M-1, \\
  &E_i^{(0)} \mapsto E_i^{(M)}, \quad 0 \leq i \leq N-2,
\end{align}
together with
\begin{align}
  E_{-1}^{(t+1)} = E_{-1}^{(t)} + \min(Q_0^{(t)}, \mathcal{E}_{0} + E_0^{(t)})
\end{align}
coincides with the rule of the hungry $\epsilon$-BBS which will be explained below.
When $M = 1$ it coincides with the $\epsilon$-BBS introduced in \cite{K} with the nonautonomous parameter $S^{(t)}$ set to $+\infty$ for all $t \in \mathbb{Z}_{\geq 0}$ (the notation $S^{(t)}$ is adopted from \cite{K}).
The hungry $\epsilon$-BBS is a discrete dynamical system on $\Omega$ with the time evolution $T_{\epsilon} \colon \Omega \to \Omega$ defined by the following:
\begin{enumerate}
  \item Set $i := 1$.
  \item For balls of color $i$, compute a time evolution of the $\epsilon$-BBS (see \cite{K} or Appendix for the rule) as if there are no balls other than the balls of color $i$.
  \item If $i = M$, then terminate. Otherwise set $i := i + 1$ and go back to Step $2$.
\end{enumerate}
\begin{ex}
  The following is an example of time evolutions of the hungry $\epsilon$-BBS for $\epsilon = (0, 1, 0, 0)$:
  \begin{align}
  t = 0:~&\us1111222\us\us\us112233\us\us\us\us133\us\us\us\us\us11223\us\us\us\us\us\us\us\us\us\us\us\us\us\us\us\us\us\us\us\us\us\us\us\us\us\us\us\us\us\us\us\us \\
  t = 1:~&\us\us\us\us\us\us\us\us11122\us\us\us1112223\us\us1333\us\us\us\us\us\us11223\us\us\us\us\us\us\us\us\us\us\us\us\us\us\us\us\us\us\us\us\us\us\us\us\us\us \\
  t = 2:~&\us\us\us\us\us\us\us\us\us\us\us\us\us11122\us112\us\us\us11223333\us\us\us11223\us\us\us\us\us\us\us\us\us\us\us\us\us\us\us\us\us\us\us\us\us\us\us\us\us \\
  t = 3:~&\us\us\us\us\us\us\us\us\us\us\us\us\us\us\us\us\us\us122\us11112\us\us\us\us\us\us\us\us1122333\us112233\us\us\us\us\us\us\us\us\us\us\us\us\us\us\us\us\us \\
  t = 4:~&\us\us\us\us\us\us\us\us\us\us\us\us\us\us\us\us\us\us\us\us\us122\us\us\us11112\us\us\us\us\us\us\us\us\us\us12233\us\us11122333\us\us\us\us\us\us\us\us\us \\
  t = 5:~&\us\us\us\us\us\us\us\us\us\us\us\us\us\us\us\us\us\us\us\us\us\us\us\us122\us\us\us\us\us11112\us\us\us\us\us\us\us\us\us\us12233\us\us\us\us\us11122333\us
  \end{align}
  The following is an example of time evolutions of the same initial sequence, but for $\epsilon = (0, 1, 1, 0)$.
  \begin{align}
    t = 0:~&\us1111222\us\us\us112233\us\us\us\us133\us\us\us\us\us11223\us\us\us\us\us\us\us\us\us\us\us\us\us\us\us\us\us\us\us\us\us\us\us\us\us\us\us\us\us\us\us\us \\
    t = 1:~&\us\us\us\us\us\us\us\us11122\us\us\us1112223\us\us1333\us\us\us\us\us\us\us11223\us\us\us\us\us\us\us\us\us\us\us\us\us\us\us\us\us\us\us\us\us\us\us\us\us \\
    t = 2:~&\us\us\us\us\us\us\us\us\us\us\us\us\us11122\us112\us\us\us11223333\us\us\us\us\us\us\us\us11223\us\us\us\us\us\us\us\us\us\us\us\us\us\us\us\us\us\us\us\us \\
    t = 3:~&\us\us\us\us\us\us\us\us\us\us\us\us\us\us\us\us\us\us122\us11112\us\us\us\us\us\us\us\us11223333\us\us\us\us\us11223\us\us\us\us\us\us\us\us\us\us\us\us\us \\
    t = 4:~&\us\us\us\us\us\us\us\us\us\us\us\us\us\us\us\us\us\us\us\us\us122\us\us\us11112\us\us\us\us\us\us\us\us\us\us\us11223333\us\us11223\us\us\us\us\us\us\us\us \\
    t = 5:~&\us\us\us\us\us\us\us\us\us\us\us\us\us\us\us\us\us\us\us\us\us\us\us\us122\us\us\us\us\us11112\us\us\us\us\us\us\us\us\us\us\us\us112233\us1122333\us\us\us
  \end{align}
\end{ex}
\section{P-symbol as a conserved quantity of the hungry $\epsilon$-BBS}
For $u \in \Omega$, let $f(u)$ denote a finite subsequence of $u$ obtained by removing all $e$'s.
For $u^{(0)}$ in the above example, we have $f(u^{(0)}) = 111122211223313311223$.
The purpose of this section is to prove the following proposition.
\begin{prop}
  For any $u \in \Omega$ and $\epsilon \in \{0, 1\}^N$, two SSTs, $\emptyset \leftarrow f(u)$ and $\emptyset \leftarrow f(T_{\epsilon}(u))$,
  coincide.
\end{prop}
Proposition 4.1 gives conserved quantities of the hungry $\epsilon$-BBS.
To prove Proposition 4.1, we use the birational transformation described in Section 3.2.
We ultradiscretize \eqref{B1}--\eqref{B3}
to obtain
\mathtoolsset{showonlyrefs=false}
\begin{align}
  &E_i^{(t,j)} = E_i^{(t, j-1)} + Q_i^{(t+M-j)} - \min(E_i^{(t, j-1)}, Q_{i+1}^{(t+M-j)}), \label{ub1} \\
  &\wt{Q}_i^{(t+M-j)} = Q_{i+1}^{(t+M-j)} + Q_i^{(t+M-j)} - \min(E_i^{(t, j-1)}, Q_{i+1}^{(t+M-j)}),  \label{ub2} \\
  &\wt{Q}_{i+1}^{(t+M-j)} = \min(E_i^{(t, j-1)}, Q_{i+1}^{(t+M-j)}).  \label{ub3}
\end{align}
\mathtoolsset{showonlyrefs=true}
We also consider transformation from $(1, \epsilon_1, \epsilon_2, ...)$-orbit to $(0, \epsilon_1, \epsilon_2, ...)$-orbit as
\begin{align}
  \wt{E}_{-1}^{(t)} = E_{-1}^{(t)} - Q_0. \label{b-1}
\end{align}
The transformation \eqref{b-1} also commute with the h-uToda orbits. We denote the (tropical) birational transformation \eqref{ub1}--\eqref{ub3} act on a pair $(\epsilon_i, \epsilon_{i+1})$ by the same symbol $\varphi_{i}$ for $i = 0, ..., N-2$,
 and define $\varphi_{-1}$ as \eqref{b-1}. The transformation $\varphi_i$ acting on $u \in \Omega_N$ is given by $\wt{\varphi}_i(u) := \Phi_N \circ \varphi_i \circ \Phi^{-1}_N(u)$.
To prove Proposition 4.1, it is sufficient to show the following:
\begin{prop}
  For all $u \in \Omega$ and $i = -1, 0, ..., N-2$, two SSTs $\emptyset \leftarrow f(u)$ and $\emptyset \leftarrow f(\wt{\varphi}_i(u))$ coincide.
\end{prop}
It is easy to see that for any $\epsilon \in \{0, 1\}^N$, there is a sequence $i_1, i_2, ..., i_k \in \{-1, 0, ..., N-2\}$ such that
$\wt{\varphi} := \wt{\varphi_{i_k}} \circ \cdots \circ \wt{\varphi}_{i_2} \circ \wt{\varphi}_{i_1}$ is the transformation from the $\epsilon$-orbit to the $\epsilon_0$-orbit where $\epsilon_0 = (0, 0, ..., 0)$.
Thus, by combining Proposition 2.4 with Proposition 4.2, we obtain Proposition 4.1. Let us prove Proposition 4.2. As the product for SSTs is associative (Proposition 2.1), it is sufficient to show that Proposition 4.2 holds for a sequence of the following form:
\begin{align}
  u = \underbrace{11...1}_{Q_0^{(0)}}~\underbrace{22...2}_{Q_0^{(1)}}~...~\underbrace{MM...M}_{Q_0^{(M-1)}}~\underbrace{ee...e}_{E_0}~\underbrace{11...1}_{Q_1^{(0)}}~\underbrace{22...2}_{Q_1^{(1)}}~...~\underbrace{MM...M}_{Q_1^{(M-1)}}~ee~... \label{seq1}
\end{align}
Let $(Q, E) = \Phi_2^{-1}(u)$.
We will prove that
\begin{align}
  \eta_n = \max_{1 \leq k \leq n} \left\{ Q_0^{(0)} + Q_0^{(1)} + \cdots + Q_0^{(k-1)} + Q_1^{(k-1)} + \cdots + Q_1^{(n-1)}  \right\}, \label{Si2}
\end{align}
for $n = 1, 2, ..., M$ is conserved under $\wt{\varphi}_0$.
First we introduce a notion of the \textit{inverse-ultradiscretization}.
The inverse-ultradiscretization is an operation of replacing $(\min, +)$ into $(+, \times)$ as
\begin{align}
  &\min(A, B) \mapsto a + b, \label{invult1} \\
  &A + B \mapsto ab. \label{invult2}
\end{align}
New variables obtained by this operation are called the \textit{geometric liftings} of the original variables.
For example, in \eqref{invult1} and \eqref{invult2}, variables $a$ and $b$ are geometric liftings of $A$ and $B$, respectively.
We perform the inverse-ultradiscretization of \eqref{Si2} to obtain
\begin{align}
  \trop^{-1}(\eta_n) = \cfrac{\prod_{i=1}^n a_i}{\sum_{j=1}^n \prod_{i = 1, i\neq j}^n a_i} \label{inv}
\end{align}
where $a_k$ is
\begin{align}
  a_k = \prod_{i=0}^{k-1} q_0^{(i)} \prod_{j=k-1}^{n-1} q_1^{(j)},
\end{align}
and $q_i^{(j)}$ and $e_0$ denote geometric liftings of $Q_i^{(j)}$ and $E_0$, respectively. We call \eqref{inv} the geometric Schensted-insertion.
We define $\beta_k, ~k = 0, 1, ..., M-1$ as
\begin{align}
  &\beta_{i} = \beta_{i-1} q_0^{(i)} + \prod_{j=0}^{i-1} q_1^{(j)},  \label{beta}\\
  &\beta_0 = 1.
\end{align}
Then \eqref{inv} is written as
\begin{align}
  \trop^{-1}(\eta_n) = \frac{q_0^{(0)}q_0^{(1)}\cdots q_0^{(n-1)} q_1^{(0)}q_1^{(1)} \cdots q_1^{(n-1)}}{\beta_{n-1}}.
\end{align}
As $q_0^{(j)}q_1^{(j)}, ~j = 0, 1, ..., N-1$, is conserved by the transformation $\varphi_0$, it suffice to show that $\beta_n$ is unchanged by the transformation $\varphi_0$.
We also define $\alpha_i,~i = 0, 1, ..., M-1$, by
\begin{align}
  &\alpha_i = \alpha_{i-1} q_1^{(M-i-1)} + e_0\prod_{j=0}^{i-1} q_0^{(M-j-1)}, \label{alpha} \\
  &\alpha_0 = e_0 + q_1^{(M-1)},
\end{align}
With $\alpha_i,~i = 0, 1, ..., M-1$, variables $\wt{q}_0^{(i)}, \wt{q}_1^{(i)}, ~i=0, 1, ..., M-1$ are written as
\begin{align}
  \wt{q}_0^{(i)} = \frac{q_1^{(i)}q_0^{(i)}\alpha_{M-2-i}}{\alpha_{M-1-i}}, \quad \wt{q}_1^{(i)} = \frac{\alpha_{M-1-i}}{\alpha_{M-2-i}} \label{q}
  \end{align}
  We use the following relation between $\alpha_i$'s and $\beta_i$'s.
  \begin{lemm}
    For $i = 0, 1, ..., M-1$, we have
    \begin{align}
      \alpha_{M-1} = \alpha_{M-1-i} \prod_{j=0}^{i-1} q_1^{(j)} + \beta_{i-1} e_0\prod_{j=i}^{M-1} q_0^{(j)}, \label{rel}
    \end{align}
    where $\beta_{-1} = 0$.
  \end{lemm}
  \begin{proof}
    We prove \eqref{rel} by induction on $i$.
    For $i = 0$, \eqref{rel} trivially holds. Suppose \eqref{rel} holds for some $i \geq 0$.
    From \eqref{alpha}, we have $\alpha_{M-1-i} = \alpha_{M-2-i}q_1^{(i)} + e_0\prod_{j=i+1}^{M-1} q_0^{(j)}$, thus
    \begin{align}
      \alpha_{M-1} &= \alpha_{M-1-i} \prod_{j=0}^{i-1} q_1^{(j)}+ \beta_{i-1} e_0 \prod_{j=i}^{M-1} q_0^{(j)}, \\
                   &= \alpha_{M-2-i} \prod_{j=0}^{i} q_1^{(j)} + (\beta_{i-1} q_0^{(i)} + \prod_{j=0}^{i-1} q_1^{(j)})e_0\prod_{j=i+1}^{M-1} q_0^{(j)} , \\
                   &= \alpha_{M-2-i} \prod_{j=0}^{i} q_1^{(j)} + \beta_i e_0 \prod_{j=i+1}^{M-1} q_0^{(j)}.
    \end{align}
    Therefore \eqref{rel} holds for $i+1$.
  \end{proof}
\begin{prop}
  $\beta_i, ~i = 0, 1, 2, ..., M-1$ is conserved by the transformation $\wt{\varphi}_0$.
\end{prop}
\begin{proof}
  We prove the assertion by induction on $i$. For $i = 0$ it is trivial. Suppose the assertion hold for $i - 1$ (that is, we have $\wt{\beta}_{i-1} = \beta_{i-1}$.).
  From \eqref{beta} and \eqref{q}, we have
  \begin{align}
    \wt{\beta}_i &= \beta_{i-1} \wt{q}_0^{(i)} + \prod_{j=0}^{i-1}\wt{q}_1^{(j)}, \\
                 &= \cfrac{\beta_{i-1}q_0^{(i)}q_1^{(i)}\alpha_{M-2-i} + \alpha_{M-1}}{\alpha_{M-i-1}} \label{eqp43}
  \end{align}
  We have $\alpha_{M-1} = \alpha_{M-1-i} \prod_{j=0}^{i-1} q_1^{(j)} + \beta_{i-1} e_0\prod_{j=i}^{M-1} q_0^{(j)}$ and $\alpha_{M-2-i}q_1^{(i)} = \alpha_{M-1-i} - e_0\prod_{j=i+1}^{M-1} q_0^{(j)}$ because of Lemma 3.3.1 and \eqref{alpha}, respectively.
  Thus, the numerator of the right-hand side of \eqref{eqp43} is transformed as
  \begin{align}
    \beta_{i-1}q_0^{(i)}q_1^{(i)}\alpha_{M-2-i} + \alpha_{M-1} &= \beta_{i-1} q_0^{(i)}(\alpha_{M-1-i} - e_0\prod_{j=i+1}^{M-1} q_0^{(j)}) + \alpha_{M-1-i}\prod_{j=0}^{i-1}q_1^{(j)} \\
    &+\beta_{i-1} e_0 \prod_{j=i}^{M-1} q_0^{(j)}, \\
                                                               &= \alpha_{M-1-i} (\beta_{i-1}q_0^{(i)} + \prod_{j=0}^{i-1} q_1^{(j)}), \\
                                                               &= \alpha_{M-1-i} \beta_{i}.
  \end{align}
  Therefore we have $\wt{\beta}_i = \beta_i$.
\end{proof}
Obviously, the inverses of transformations \eqref{ub1}--\eqref{ub3} also preserve the P-symbol.
Thus, together with the remark stated in the last sentence of Section 3.2, we also see that Proposition 4.1 follows without going through Proposition 2.4.
\section{Concluding remarks}
In this paper, we first introduced the discrete hungry elementary Toda orbits and derived the hungry $\epsilon$-BBS.
The hungry $\epsilon$-BBS contains Takahashi-Satsuma's BBS with several kind of balls as a special case.
Next, we proved that the birational transformation among the elementary Toda orbits introduced in \cite{FG1} commutes with the time evolution of the discrete hungry elementary Toda orbits.
This transformation and its inverse were written without subtraction, which is important for the proof of the main theorem of this paper.
Finally, we showed that the P-symbol of the RSK correspondence is a conserved quantity of the hungry $\epsilon$-BBS, thus generalizing the earlier work by Fukuda \cite{F}.
This follows from the fact that the birational transformation \eqref{B1}--\eqref{B5} preserves the image of geometric Schensted insertion by Noumi and Yamada \cite{NY}.

There are several problems left for future work. The linearization of the generalized $\epsilon$-BBS is possible in principle by combining
the birational transformation \eqref{B1}--\eqref{B5} and the rigged configuration map of type $A_n^{(1)}$ (see \cite{KOSTY} for the linearization of the $A_n^{(1)}$ automata). However we have not yet written
the composition of those maps explicitly in combinatorial terms.
In addition, the crystal-theoretic interpretation of the transformation \eqref{B1}--\eqref{B5} remains to be investigated.
Another natural question is whether similar extensions are possible for the BBS other than the type $A_n^{(1)}$.
\section*{Acknowledgements}
The research of KK was partially supported by Grant-in-Aid for JSPS  Fellows, 19J23445. The research of ST was partially supported by JSPS Grant-in-Aid for Scientific Research (B), 19H01792. This research was partially supported by the joint project ``Advanced Mathematical Science for Mobility Society'' of
Kyoto University and Toyota Motor Corporation.
\section*{Appendix: Time evolution rule of the $\epsilon$-BBS}
In this section, we present a time evolution rule of the $\epsilon$-BBS.
For a $01$-sequence $u = (u_i)_{i=0}^{\infty}$, we define integers $Q_0, Q_1, ..., Q_{N-1}$ and $E_{-1}, E_0, ..., E_{N-2}$ as follows:
\begin{itemize}
\item $Q_i$: the length of the $(i+1)$-st block of consecutive balls in $u$ and
\item $E_i$: the number of empty boxes between the $(i+1)$-st and the $(i+2)$-nd blocks of consecutive balls in $u$.
\end{itemize}
\begin{ex}
  For a sequence $u = 0111011001100111001110110000000000000\cdots$, we define $Q_0, Q_1, ..., Q_{N-1}$ and $E_{-1}, E_0, ..., E_{N-2}$ as follows:
  \begin{align}
    &Q_0 = 3, \quad Q_1 = 2, \quad  Q_2 = 2, \quad   Q_3 = 3, \quad   Q_4 = 3, \quad   Q_5 = 2, \\
    &E_{-1} = 1, \quad E_0 = 1, \quad E_1 = 2, \quad  E_2 = 2, \quad   E_3 = 2, \quad   E_4 = 1.
  \end{align}
\end{ex}
Let us explain the rule of the time evolution $T_{\epsilon} \colon u \mapsto T_{\epsilon}(u)$ of the $\epsilon$-BBS for a given $\epsilon \in \{0, 1\}$ using the sequence $u$ of Example 5.1 as an example.
First we prepare relevant notations.
Let $I = \{ i \in \{0, ..., N-1\} \mid \epsilon_i = 1\}$ and
elements of $I$ be $I = \{ i_0, i_1, ..., i_{K-1} \mid i_0 < i_1 < \cdots < i_{K-1} \}$, where $K = |I|$.
We set $i_{-1} = 0$ and
let $m_j = k_{2i_j}$ for $j = 0, 1, ..., K-1$.
We decompose the sequence $u$ into subsequences $v^{(j)} = (u_{m_{j-1}}, u_{m_{j-1}+1}, ..., u_{m_{j}-1})$ for $j = 0, 1, ..., K$, where
$m_{-1} = 0$ and $m_K = +\infty$. In the $01$-sequence $u$ of Example 5.1 and $\epsilon = (1, 0, 1, 0, 1, 0)$, we have $i_0 = 0, i_1 = 2$ and $i_2 = 4$. Hence, $m_0 = 1, m_1 = 9$ and $m_2 = 18$ and the decomposition of a sequence $u^{(0)}$ is
$v^{(0)} = (0)$, $v^{(1)} = (1,1,1,0,1,1,0,0), v^{(2)} = (1,1,0,0,1,1,1,0,0)$, and $v^{(3)} = (1,1,1,0,1,1,0,0,0,0,...)$.

We explain the rule of the $\epsilon$-BBS in terms of a carrier that moves from left to right.
Let $c^{(-1)} = 0$. We construct a map that takes $v^{(j)}$ and $c^{(j-1)}$ as inputs and outputs a $01$-sequence $\wt{v}^{(j)}$ and a nonnegative integer $c^{(j)}$ for $j = 0, 1, ..., K$.

First, we start with a carrier containing $c^{(j-1)}$ balls, and move the carrier from left to right until it reaches the right end of $v^{(j)}$. As the carrier passes each position, perform one of the following:
\begin{itemize}
  \item When the carrier comes across a ball, load it onto the carrier.
  \item When the carrier comes across an empty box and contains no ball, do nothing.
  \item When the carrier comes across an empty box and contains at least one ball, unload a ball.
  However, when unloading a ball for the first time in step $j$, remove $c^{(j-1)}$ balls from the carrier (this procedure is indicated by the double-lined arrow in diagrams \eqref{arrow1} and \eqref{arrow2} in Example 3.2 below).
\end{itemize}
Then, we obtain the finite $10$-sequence $(v')^{(j)}$ and the carrier contents $c^{(j)}$.
Next, we add $c^{(j-1)}$ balls into the first (leftmost) block of balls of $(v')^{(j)}$ and, if $j > 0$, delete $\max(Q_{i_{j-1}}- E_{i_{j-1}}, 0)$ boxes from the first (leftmost) block of empty boxes of $(v')^{(j)}$. We
define $\wt{v}^{(j)}$ by the resulting sequence.

After executing the above procedures for $j = 0, 1, ..., K$, we concatenate sequences $\wt{v}^{(0)}, \wt{v}^{(1)}, ..., \wt{v}^{(K)}$ to obtain $T_{\epsilon}(u) = \wt{v}^{(0)}\wt{v}^{(1)}\cdots\wt{v}^{(K)}$.

\begin{ex}
Let us give an example of the above rule for the $01$-sequence $u$ in Example 5.1 for $\epsilon = (1, 0, 1, 0, 1, 0)$.
  First, we explain the above procedure for $j = 1$ (in the case of $j = 0$, we trivially obtain $\wt{v}^{(0)} = (0)$).
  Let $c$ and $\wt{c}$ be respectively the states of the carrier before and after it passes through position $v^{(1)}_l$.
   The following diagram illustrates the changes in the state of the carrier and $v^{(1)}_l$ of the $01$-sequence before and after the carrier passes:
  \begin{align}
    &
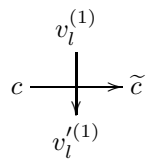
\xymatrix@C=5pt@R=5pt{
    &v^{(1)}_l\ar[dd]& \\
    c \ar[rr]&&\wt{c}  \\
    &v'^{(1)}_l&
    }
  \end{align}
The diagram below shows how the state of the carrier changes as it moves from $u_{m_0}$ to $u_{m_1-1}$:
  \begin{align}
    &
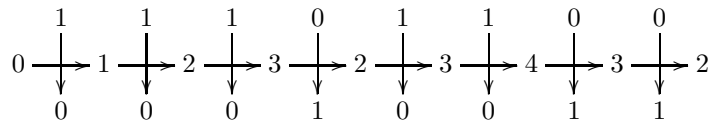
\xymatrix@C=5pt@R=5pt{
    &1\ar[dd]&&1\ar[dd]&&1\ar[dd]&&0\ar[dd]&&1\ar[dd]&&1\ar[dd]&&0\ar[dd]&&0\ar[dd] \\
    0\ar[rr]&&1 \ar[rr] && 2\ar[rr]&&3 \ar[rr]&&2 \ar[rr]&&3\ar[rr]&&4\ar[rr] &&3\ar[rr] && 2 \\
    &0&&0&&0&&1&&0&&0&&1&&1
    }
  \end{align}
  After the carrier passes $u_{m_1-1}$, we obtain the sequence $v'^{(1)} = (0, 0, 0, 1, 0, 0, 1, 1)$ and the carrier contents $c^{(1)} = 2$. Then we
  delete $\max(3 - 1, 0) = 2$ empty boxes from $v'^{(1)}$ and obtain the resulting sequence $\wt{v}^{(1)} = (0, 1, 0, 0, 1, 1)$.

  Next, let us consider the cases $j = 2$ and $j = 3$.
  When $j = 2$, we have following diagram:
  \begin{align}
    \xymatrix@C=5pt@R=5pt{
    &1\ar[dd]&&1\ar[dd]&&0\ar[dd]&&0\ar[dd]&&1\ar[dd]&&1\ar[dd]&&1\ar[dd]&&0\ar[dd]&&0\ar[dd]& \\
    2\ar[rr]&&3 \ar[rr] && 4\ar@{=>}[rr]&&1 \ar[rr]&&0 \ar[rr]&&1\ar[rr]&&2\ar[rr] &&3\ar[rr] && 2\ar[rr]&& 1 \\
    &0&&0&&1&&1&&0&&0&&0&&1&&1&
    }\\\label{arrow1}
  \end{align}
  Therefore we obtain $(v')^{(2)} = (0, 0, 1, 1, 0, 0, 0, 1, 1)$ and $c^{(2)} = 1$. Because $c^{(1)} = 2$ and $\max(2 - 2, 0) = 0$, we have $\wt{v}^{(2)} = (0, 0, 1, 1, 1, 1, 0, 0, 0, 1, 1)$.
  When $j = 3$, we have the following diagram:
  \begin{align}
    &\xymatrix@C=5pt@R=5pt{
    &1\ar[dd]&&1\ar[dd]&&1\ar[dd]&&0\ar[dd]&&1\ar[dd]&&1\ar[dd]&&0\ar[dd]&&0\ar[dd]&&0\ar[dd]&&0\ar[dd]& \\
    1\ar[rr]&&2 \ar[rr] && 3\ar[rr]&&4 \ar@{=>}[rr]&&2 \ar[rr]&&3\ar[rr]&&4\ar[rr] &&3\ar[rr] && 2\ar[rr]&& 1\ar[rr]&& 0 \\
    &0&&0&&0&&1&&0&&0&&1&&1&&1&&1&
    }\\\label{arrow2}
  \end{align}
  Therefore we obtain $(v')^{(3)} = (0, 0, 0, 1, 0, 0, 1, 1, 1, 1, 0, ...)$ and $c^{(3)} = 0$. Because $c^{(2)} = 1$ and $\max(3 - 1, 0) = 2$, we have $\wt{v}^{(3)} = (0, 1, 1, 0, 0, 1, 1, 1, 1, 0, ...)$. Finally, by concatenating sequences
  \begin{align}
    &\wt{v}^{(0)} = (0), \quad \wt{v}^{(1)} = (0, 1, 0, 0, 1, 1), \quad \wt{v}^{(2)} = (0, 0, 1, 1, 1, 1, 0, 0, 0, 1, 1), \\
    &\wt{v}^{(3)} = (0, 1, 1, 0, 0, 1, 1, 1, 1, 0, 0, 0, 0, ...),
  \end{align}
  we obtain $T_{\epsilon}(u) = ``001001100111100011011001111000..."$. \qed
\end{ex}

\end{document}